\documentclass[pra,twocolumn,showpacs]{revtex4-2}

\usepackage{amsmath}
\usepackage{latexsym}
\usepackage{amssymb}
\usepackage{graphicx}
\usepackage[colorlinks=true, citecolor=blue, urlcolor=blue]{hyperref}
\usepackage{float}
\usepackage{amsfonts}
\usepackage{textcomp}
\usepackage{mathpazo}
\usepackage{comment}
\usepackage{xr}

\usepackage{bbm}

\usepackage{xcolor}
\definecolor{myurlcolor}{rgb}{0,0,0.4}
\definecolor{mycitecolor}{rgb}{0,0.5,0}
\definecolor{myrefcolor}{rgb}{0.5,0,0}
\usepackage{hyperref}
\hypersetup{colorlinks,
linkcolor=myrefcolor,
citecolor=mycitecolor,
urlcolor=myurlcolor}

\usepackage[draft]{fixme}
\usepackage{amsmath,bbm}
\usepackage{graphicx}
\usepackage{amsfonts}
\usepackage{amssymb}
\usepackage{amsmath, amssymb, amsthm,verbatim,graphicx,bbm}
\usepackage{mathrsfs}
\usepackage{color,xcolor,longtable}

\usepackage{xr}
\makeatletter
\newcommand*{\addFileDependency}[1]{
  \typeout{(#1)}
  \@addtofilelist{#1}
  \IfFileExists{#1}{}{\typeout{No file #1.}}
}
\makeatother

\newcommand*{\myexternaldocument}[1]{
    \externaldocument{#1}
    \addFileDependency{#1.tex}
    \addFileDependency{#1.aux}
}

\myexternaldocument{manuscript_PRL}

\newcommand{\beq}[0]{\begin{equation}}
\newcommand{\eeq}[0]{\end{equation}}

\newcommand{\one}{\leavevmode\hbox{\small1\normalsize\kern-.33em1}}

\def\be{\begin{equation}}
\def\ee{\end{equation}}
\def\ben{\begin{eqnarray}}
\def\een{\end{eqnarray}}
\def\eea{\end{array}}
\def\bea{\begin{array}}

\newcommand{\Tr}[1]{\mathrm{Tr}#1}
\newcommand{\bei}{\begin{itemize}}
\newcommand{\eei}{\end{itemize}}
\newcommand{\ket}[1]{|#1\rangle}
\newcommand{\bra}[1]{\langle#1|}

\newcommand{\proj}[1]{\ket{#1}\!\!\bra{#1}}

\newcommand{\I}{\mathbbm{1}}

\renewcommand{\emph}[1]{\textbf{#1}}

\makeatletter
\newtheorem*{rep@theorem}{\rep@title}
\newcommand{\newreptheorem}[2]{%
\newenvironment{rep#1}[1]{%
 \def\rep@title{#2 \ref{##1}}%
 \begin{rep@theorem}}%
 {\end{rep@theorem}}}
\makeatother

\theoremstyle{plain}
\newtheorem{thm}{Theorem}
\newtheorem*{thm*}{Theorem}
\newreptheorem{thm}{Theorem}

\newtheorem{fakt}{Fact}
\newtheorem{ass}{Assumption}

\newtheorem{defn}[thm]{Definition}

\theoremstyle{definition}

\theoremstyle{remark}

\usepackage[T1]{fontenc}

\begin{document}

\title{Certification of unbounded randomness with arbitrary noise}
\author{Shubhayan Sarkar}
\email{shubhayan.sarkar@ulb.be}
\affiliation{Laboratoire d’Information Quantique, Université libre de Bruxelles (ULB), Av. F. D. Roosevelt 50, 1050 Bruxelles, Belgium}

\begin{abstract}	
Random number generators play an essential role in cryptography and key distribution. It is thus important to verify whether the random numbers generated from these devices are genuine and unpredictable by any adversary. Recently, quantum nonlocality has been identified as a resource that can be utilised to certify randomness. Although these schemes are device-independent and thus highly secure, the observation of quantum nonlocality is extremely difficult from a practical perspective. In this work, we provide a scheme to certify unbounded randomness in a semi-device-independent way based on the maximal violation of Leggett-Garg inequalities. Interestingly, the scheme is independent of the choice of the quantum state, and consequently even classical noise like a thermal state or even microwave background radiation could be utilized to self-test quantum measurements and generate unbounded randomness making the scheme highly efficient for practical purposes.
\end{abstract}


\maketitle

{\it{Introduction---}}
Random numbers play a crucial role in cryptography and key distribution, serving as a fundamental ingredient for ensuring the security and confidentiality of sensitive information. These classical random number generators are based on the limited knowledge of the physical process that generates these numbers. Consequently, one needs to trust that the knowledge of the process is completely hidden from any adversary who might have access to these devices. The randomness of such numbers is thus certified in a device-dependent way.

Unlike classical physics where in principle events are determined with certainty, quantum theory describes the behavior of particles and systems in terms of probabilities. Further on, the unpredictability of measurement outcomes in quantum theory is intrinsic and not due to ignorance, thus serving as an excellent tool for generating random numbers. In recent times, the concept of quantum non-locality, manifested by the violation of Bell inequalities \cite{Bell}, has emerged as a means to certify randomness in a device-independent (DI) manner \cite{di4,AM16}. This implies that the assessment of randomness is decoupled from the specific physical characteristics or details of the experimental setup. There are several schemes that utilize quantum nonlocality for DI certification of randomness \cite{colbeck2011, random0, pironio222, random1, remik16,remik17, sarkar,sarkar5,Armin1}.

However, from a practical perspective, observation of quantum non-locality in a loophole-free way is an extremely difficult task. All of these experiments are highly sensitive to noise and require highly entangled sources which is a costly resource \cite{Bellexp1, Bellexp2, Bellexp3, Bellexp4}. Furthermore, the device-independent randomness generation schemes suffer from low rates and are highly sensitive to detector noise \cite{DIrand1, DIrand2, DIrand3, DIrand4} and thus highly demanding from a practical perspective. As a consequence, it is worth exploring scenarios that are noise-resistant and easy to implement. In this regard, some physically well-motivated assumptions can be made on the devices which do not compromise much over security but are easier to implement. Such schemes are known as semi-device-independent (SDI). One such assumption is that one of the parties involved in the experiment is fully trusted, that is, the measurements performed by the trusted party are known. Such schemes are considered to be one-sided device-independent (ISDI) \cite{steerand2,steerand3,steerand4,sarkar12,Sarkar_2024}. In particular, Ref. \cite{sarkar12} proposes a 1SDI scheme to certify the optimal randomness from measurements with arbitrary number outcomes. 

In this work, we consider a sequential scenario inspired by Leggett-Garg (LG) inequalities \cite{LG} where a single system is measured in a "time-like" separated way. Any violation of LG inequality implies that quantum theory violates the notion of "memoryless" hidden variable models, which as a matter of fact can also be violated in classical physics. For instance, even a classical pre-programmed device can reproduce any observed correlations in the sequential scenario as the device might have a record of the previous inputs and outputs. Consequently, an assumption that we impose in this work is that the correlations obtained in the experiment are generated by input-consistent measurements acting on some quantum state making the proposed scheme semi-device-independent. For our purpose, we consider the generalized LG inequality with arbitrary number of inputs \cite{lg2} and self-test qubit measurements spanning the entire $X-Z$ plane up to the presence of local unitaries. For a note, self-testing of quantum measurements using the LG inequalities for the particular case of four inputs was proposed in \cite{Jeba} and its generalization to arbitrary number of outcomes was proposed in \cite{Das} that assumed a particular form of the initial quantum state. Then, we utilise the certified measurements to certify unbounded amount of randomness from the untrusted devices. 
 
 A scheme proposed in \cite{remik17} also utilises sequential measurements for generating unbounded randomness. However, it is based on violation of Bell inequalities which is again difficult to observe. 
 Interestingly, the scheme presented in this work is independent of the initial quantum state and thus even classical noise can be used to generate unbounded randomness. To the best of our knowledge, this is the first scheme that can be used to generate unbounded randomness in a state-independent way. Further on, violation of LG inequalities have been observed in a large number of quantum systems \cite{LGexp1, LGexp2, LGexp3,LGexp4,LGexp5}, thus making our scheme an excellent candidate for practical random number generators.

{\it{Sequential scenario---}}\label{sec2}
The sequential scenario consists of a source and a measurement device with $n-$inputs labeled as $x=1,2,\ldots,n$ and binary outcomes labeled as $a=0,1$. 
Now in a single run of the experiment, the user provides an arbitrary number of inputs in a sequential manner (one after another) to the device and records their outcomes. From the experiment one can obtain the distribution $\vec{p}_N=\{p(a_1,a_2,\ldots,a_N|x_1,x_2,\ldots,x_N)\}$ where $N$ is the number of consecutive inputs and $p(a_1,a_2,\ldots,a_N|x_1,x_2,\ldots,x_N)$ signifies the probability of obtaining outcomes $a_1,a_2,\ldots,a_N$ consequetively when one inputs $x_1,x_2,\ldots,x_N$ to the device [see Fig. \ref{fig1}]. 

Using the above set-up Leggett and Garg proposed a test referred to as "Leggett-Garg (LG)" inequality that allows one to exclude macrorealist non-invasive description of quantum theory [for detailed analysis refer to \cite{LGreview}]. The LG inequality is given by
\begin{eqnarray}\label{lg1}
    \mathcal{L}&=&\sum_{x=1}^{n-1}C_{x,x+1}-C_{n,1}\leq \beta_{\mathcal{M}}(n)
\end{eqnarray}
where the terms $C_{x,y}$ represent the two-time correlation between the inputs $x,y$ and can be obtained via $\vec{p}_2$ as
\begin{eqnarray}
    C_{x,y}=\sum_{a_1,a_2}(-1)^{a_1+a_2}p(a_1,a_2|x,y).
\end{eqnarray}
The above correlation can be generalized to an arbitrary number of sequential measurements $C_{x_1,x_2,\ldots,x_N}$ as
\begin{equation}\label{gencor}
   C_{x_1,\ldots,x_N}=\sum_{a_1,\ldots,a_N}(-1)^{a_1+\ldots+a_N}p(a_1,\ldots,a_N|x_1,\ldots,x_N).
\end{equation}
In the inequality \eqref{lg1}, $\beta_{\mathcal{M}}(n)$ denotes the maximum value that one can achieve when the distribution $\vec{p}_2$ can be expressed via "time-local" or "memory-less" hidden variable models given as
\begin{eqnarray}
p(a_1,a_2|x,y)=\sum_{\lambda}p(a_1|x,\lambda)p(a_2|y,\lambda)p(\lambda). 
\end{eqnarray}
with the value $\beta_{\mathcal{M}}(n)=n-2$.

\begin{figure}
    \centering
    \includegraphics[scale=.8]{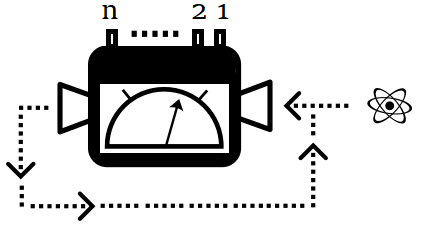}
    \caption{The sequential scenario. The source sends a single system into the measurement device with $n$ inputs labelled as $x_i=1,2,\ldots,n$ and binary outcomes labelled as $a_i=0,1$ with $i=1,\ldots,N$ denoting the sequence of measurements. The quantum state is measured in sequential way to obtain the probability distribution $\vec{p}_N$.}
    \label{fig1}
\end{figure}

Let us now restrict ourselves to quantum theory where each input $i$ corresponds to a fixed measurement $A_x=\{\mathbbm{M}_{x,0},\mathbbm{M}_{x,1}\}$ where $\mathbbm{M}_{x,j}$ represent measurement elements that are positive and $\sum_j\mathbbm{M}_{x,j}=\I$. The measurement elements in general are not projective. Consequently, the corresponding probability $p(a_1,a_2|A_1,A_2)$ is given by
\begin{equation}\label{probgen}
 p(a_1,a_2|A_1,A_2)=
    \Tr\left(\sqrt{\mathbbm{M}_{1,a_1}}\ U_{a_1}^{\dagger}\mathbbm{M}_{2,a_2}U_{a_1}\sqrt{\mathbbm{M}_{1,a_1}}\ \rho_A\right)
\end{equation}
where $U_{a_1}$ is some unitary dependent on the outcome $a_1$ and $\rho_A$ is some quantum state. The above rule to compute probability can be straightaway generalised to an arbitrary number of sequential measurements.

 Let us now consider that the measurements $A_i$ corresponding to each input $i$ are projective. As pointed out by Fritz in \cite{Fritz_2010} for projective measurements, the correlation $C_{x,y}$ in quantum theory is expressed as $C_{x,y}=1/2\left\langle\{\mathcal{A}_x,\mathcal{A}_y\}\right\rangle$ where $\langle O\rangle=\Tr(O\rho)$ for some operator $O$ and $\mathcal{A}_x$ denotes the observable corresponding to the $x-th$ measurement represented in terms of the measurement elements $\Pi_{x,j}\ (j=0,1)$ as 
\begin{eqnarray}
    \mathcal{A}_x=\Pi_{x,0}-\Pi_{x,1}.
\end{eqnarray}
It is simple to observe that  $\mathcal{A}_x^2=\I$. Consequently, $p(a_1,a_2,\ldots,a_N|x_1,x_2,\ldots,x_N)$ is expressed for projective measurements as
\begin{eqnarray}
p(a_1,a_2,\ldots,a_N|x_1,x_2,\ldots,x_N)=\qquad\qquad\qquad\qquad\nonumber\\
    \Tr\left(\Pi_{x_{1},a_{1}}\ldots\Pi_{x_{N-1},a_{N-1}}\Pi_{x_N,a_N}\Pi_{x_{N-1},a_{N-1}}\ldots\Pi_{x_{1},a_{1}}\rho\right)\nonumber\\
\end{eqnarray}
Thus, for projective measurements in quantum theory the witness $ \mathcal{L}$ from \eqref{lg1} is given by
\begin{eqnarray}\label{lg2}
     \mathcal{L}&=&\frac{1}{2}\sum_{x=1}^{n-1}\left\langle\{\mathcal{A}_x,\mathcal{A}_{x+1}\}\right\rangle-\frac{1}{2}\left\langle\{\mathcal{A}_n,\mathcal{A}_1\}\right\rangle.
\end{eqnarray}
Consider now the following observables
\begin{eqnarray}\label{obs1}
    \tilde{\mathcal{A}}_x=\cos{\frac{\pi (x-1)}{n}}\sigma_z+\sin{\frac{\pi (x-1)}{n}}\sigma_x
\end{eqnarray}
where $\sigma_z, \sigma_x$ are the Pauli $z,x$ matrices. Now, a simple computation of the functional \eqref{lg2} using the observables \eqref{obs1} yields the value $\beta_Q(n)=n\cos{\frac{\pi}{n}}$ which is strictly greater than $\beta_{\mathcal{M}}(n)$. We will show later that $\beta_Q(n)$ is in fact the maximum value of $\mathcal{L}$ attainable using quantum theory when restricting to projective measurements. 

Before proceeding, let us recall an important constraint that is imposed on the distribution $\vec{p}_N$ known as "no-signalling in time" \cite{LGreview} conditions given by
\begin{eqnarray}\label{NSIT}
    \sum_{\substack{i=1\\i\ne k}}^N\sum_{a_i=0,1}p(a_1,a_2,\ldots,a_N|x_1,x_2,\ldots,x_N)=p(a_k|x_k)
\end{eqnarray}
for any $x_1,\ldots,x_N$. Before proceeding to the main results, let us now comment on whether the above-described sequential scenario can be utilised for device-independent quantum information or not.


{\it{Self-testing quantum measurements in a state-independent way---}}
Self-testing is a method of DI certification where one can characterize the quantum states and measurements inside an untrusted device up to some degree of freedom under which the observed probabilities remain invariant. In this section, we self-test any qubit measurement in the $X-Z$ plane. To begin with, let us clearly state the major assumption that is imposed in the sequential scenario for obtaining the self-testing result.

\begin{ass}[Input-consistent measurements]\label{ass1} The correlations $\vec{p}_N$ obtained in the sequential scenario [see Fig. \ref{fig1}] are generated by 
measurements acting on some state that are consistent for a particular input. 
\end{ass}

The consistency of measurements for a particular input ensures that they are independent of any previous input-output. This allows us to consider that $A_i's$ are POVM's as discussed in Eq. \eqref{probgen}. Let us now revisit the previous experiment [see Fig. \ref{fig1}] in which a user sequentially measures a quantum state $\rho_A$ sent by the source and observes the correlations $\vec{p}_N$. Consider now a reference experiment that reproduces the same statistics as the actual experiment but involves the states $\tilde{\rho}_{A}$ and observables represented by $\tilde{\mathcal{A}}_i$. The observables $\mathcal{A}_i$ are self-tested from $\vec{p}_N$ if there exists a unitary $\mathcal{U}:\mathcal{H}_A\to \mathcal{H}_{A'}\otimes\mathcal{H}_{A''}$ such that 
\begin{equation}
\mathcal{U}\mathcal{A}_i\mathcal{U}^{\dagger}=\tilde{\mathcal{A}}_i\otimes\mathbbm{1}_{A''},
\end{equation}
%
%
where $\mathcal{H}_{A''}$ denotes the junk Hilbert space and $\mathbbm{1}_{A''}$ denotes the identity acting on $\mathcal{H}_{A''}$. The self-testing result presented in this work is state-independent and consequently no state can be certified using our scheme. Before proceeding, let us recall that the observables can be certified on the support of the quantum state. Thus without loss of generality throughout the manuscript, we will assume that the quantum state $\rho_A$ is full-rank.

Let us now restrict ourselves to the probability distribution $\vec{p}_2$. Inspired by \cite{Das}, we impose the following condition on $\vec{p}_2$.
\begin{defn}[Zeno conditions]\label{def1} If the same measurement $A_i$ for any $i$ is performed sequentially, then for both measurement events the same outcome occurs with certainty.  This implies that the distribution $\vec{p}_2$ is constrained as
\begin{eqnarray}\label{zenocond}
    p(a,b|A_i,A_i)=\delta_{a,b}p(a|A_i) \qquad \forall a,b,i.
\end{eqnarray}
    
\end{defn}

Let us note that the above condition is operational and one can verify it from the statistics generated in the experiment by successively performing the same measurement.
Using assumption \ref{ass1}, we show in fact 1 in Appendix A of \cite{SupMat}, that the condition \eqref{zenocond} implies that the measurements $A_i$ are projective. This allows us to consider the Leggett-Garg functional \eqref{lg2}. 
Let us show that $\beta_Q(n)$ is the maximal quantum value of $\mathcal{L}$ \eqref{lg2}. For this purpose, we consider the  LG operator $\hat{\mathcal{L}}$ given by
\begin{eqnarray}
    \hat{\mathcal{L}}&=&\frac{1}{2}\sum_{x=1}^{n-1}\{\mathcal{A}_x,\mathcal{A}_{x+1}\}-\frac{1}{2}\{\mathcal{A}_n,\mathcal{A}_1\}.
\end{eqnarray}
Consider now the following operators $P_i$ for $i=1,\ldots,n-2$ given by
\begin{eqnarray}\label{P}
    P_i=\mathcal{A}_i-\alpha_ i \mathcal{A}_{i+1}+\beta_ i \mathcal{A}_{n}
\end{eqnarray}
where 
\begin{eqnarray}\label{alpha}
    \alpha_i=\frac{\sin\left(\frac{\pi i}{n}\right)}{\sin\left(\frac{\pi (i+1)}{n}\right)},\qquad \beta_i=\frac{\sin\left(\frac{\pi}{n}\right)}{\sin\left(\frac{\pi (i+1)}{n}\right)}.
\end{eqnarray}
After some simplification, one can observe that
\begin{equation}\label{sos1}
    \sum_{i=1}^{n-2}\frac{1}{2\alpha_i}P_i^{\dagger}P_i=\frac{1}{2}\sum_{i=1}^{n-2}\left(\frac{1}{\alpha_i}+\alpha_i+\frac{\beta_i^2}{\alpha_i}\right)\I-\hat{\mathcal{L}}
\end{equation}
where we used the fact that $\mathcal{A}_i^2=\I$.
Notice that the term on the left-hand side of the above formula is positive which allows us to conclude that
\begin{eqnarray}\label{sos2}
   \hat{\mathcal{L}}\leq \frac{1}{2}\sum_{i=1}^{n-2}\left(\frac{1}{\alpha_i}+\alpha_i+\frac{\beta_i^2}{\alpha_i}\right)\I
\end{eqnarray}
In Fact 2 in the Appendix B of \cite{SupMat}, we show that
\begin{eqnarray}
    \sum_{i=1}^{n-2}\left(\frac{1}{\alpha_i}+\alpha_i+\frac{\beta_i^2}{\alpha_i}\right)=2\beta_Q(n)
\end{eqnarray}
which allows us to infer from \eqref{sos1} that 
\begin{eqnarray}
    \hat{\mathcal{L}}\leq\beta_Q(n)\I.
\end{eqnarray}
Consequently, $ \beta_Q(n)$ is the maximal quantum value of $\mathcal{L}$ \eqref{lg2}. 

Now, let us assume that one observes the value $ \beta_Q(n)$ of the LG functional $\mathcal{L}$ \eqref{lg2}. Thus from the decomposition \eqref{sos1}, we have that 
\begin{eqnarray}\label{sosrel2}
    \Tr(P_i^{\dagger}P_i\rho_A)=0, \qquad i=1,\ldots,n-2.
\end{eqnarray}
The above relation \eqref{sosrel2} will be particularly useful for self-testing as stated below.

\setcounter{thm}{0}
\begin{thm}\label{Theo1M} 
Assume that the Zeno conditions \eqref{zenocond} are satisfied and the LG inequality \eqref{lg1} is maximally violated by some state $\rho_A$ and observables $\mathcal{A}_i\ (i=1,\ldots,n)$. Then, the following statements hold true:
\\
\\
1. The observables $\mathcal{A}_i$ act on the Hilbert space  $\mathcal{H}_{A}=(\mathbbm{C}^2)_{A'}\otimes \mathcal{H}_{A''}$ for some auxiliary Hilbert space $\mathcal{H}_{A''}$.\\
\\
2.  \ \  There exist a unitary transformation, $\mathcal{U}:\mathcal{H}_A\rightarrow\mathcal{H}_A$,  such that
\begin{eqnarray}\label{lem1.2}
\mathcal{U}\mathcal{A}_i\mathcal{U}^{\dagger}=\tilde{\mathcal{A}}_i\otimes\mathbbm{1}_{A''}.
\end{eqnarray}
where the observables $\tilde{\mathcal{A}}_i$ are listed in Eq. \eqref{obs1}.
\end{thm}

The proof of the above theorem is given in Appendix C of \cite{SupMat}.
Interestingly, the above self-testing result is valid for any quantum state. Just like any other self-testing scheme, we can always consider that the input state is full-rank. This is because any correlation that one obtains in an experiment is only via some measurements acting on the support of the state. So every measurement can only be certified only on the support of the state and thus it is equivalent to assuming that the input state is full-rank. 

From a practical perspective, one can never exactly prepare the measurements to obtain the exact maximal value of the LG inequality \eqref{lg2}. Assuming that one can prepare projective measurements and thus satisfy the Zeno conditions def \ref{def1}, we find the violation of the LG inequality \eqref{lg2} to be robust as stated below.

\begin{thm}
    Suppose that the observables in the actual experiment are close to the ideal ones as
    \begin{eqnarray}\label{obserr}
    ||(\mathcal{A}_i-\mathcal{A}'_i)\sqrt{\rho_A}||\leq \varepsilon
    \end{eqnarray}
where $\mathcal{A}'_i=\mathcal{U}\left(\tilde{\mathcal{A}}_i\otimes\I\right) \mathcal{U}^{\dagger}$ and $\tilde{\mathcal{A}}_i$ are listed in Eq. \eqref{obs1}. Here $\rho_A$ is the actual state during the experiment. Then, the LG inequality \eqref{lg2} is violated close to the quantum bound as
\begin{eqnarray}
    \mathcal{L}\geq \beta_Q(n)-\frac{n(1+2\cos(\pi/n))}{2}\varepsilon.
\end{eqnarray}
\end{thm}

The proof of the above theorem can be found in Appendix D  of \cite{SupMat}.

Let us now utilize the above self-testing result in the noiseless scenario to certify unbounded amount of randomness generated from the untrusted measurements.

{\it{State-independent unbounded randomness expansion---}}
Here we certify unbounded randomness from the untrusted measurements in the sequential scenario. For this purpose, we first consider assumption \ref{ass1} along with the Zeno conditions \eqref{zenocond} which ensures that the measurements are projective.
Let us now restrict to even $n$ and consider the correlation $C_{i,i+n/2,i,i+n/2,\ldots}$ for any $i$ such that $(i=2,\ldots,\frac{n}{2})$ corresponding to the distribution when the observables $\mathcal{A}_i,\mathcal{A}_{i+n/2}$ are sequentially measured. In terms of probabilities, the correlation $C_{i,i+n/2,i,i+n/2,\ldots}$ is expressed in Eq. \eqref{gencor}. Consequently, we modify the LG inequality as
\begin{eqnarray}\label{R}
    \mathcal{R}_i=\mathcal{L}-|C_{i,i+n/2,i,i+n/2,\ldots}| \quad i=2,\ldots,\frac{n}{2}.
\end{eqnarray}
Notice that using the observables listed in \eqref{obs1}, one can attain the value $\beta_Q(n)=n\cos(\frac{\pi}{n})$ of $\mathcal{R}_i$ for any $i$. As $\mathcal{R}_i\leq\mathcal{L}$, it is thus clear that the maximum quantum value of $\mathcal{R}_i$ is the same as $\mathcal{L}$. Now, if one observes the maximal quantum value $\beta_Q(n)$ of $\mathcal{R}_i$, then $|C_{i,i+n/2,i,i+n/2,\ldots}|=0$ and $\mathcal{L}=\beta_Q(n)$. Thus, from theorem \ref{Theo1M}, we can conclude that the observables $\mathcal{A}_i$ are certified as in \eqref{lem1.2}. 

Now, let us compute the guessing probability of an adversary Eve who might have access to the user's quantum state. The joint state of Eve and the user is denoted as $\rho_{AE}$ such that $\rho_A=\Tr_E(\rho_{AE})$. As Eve's dimension is unrestricted, without loss to generality we assume that $\rho_{AE}$ is pure and denote it further as $\psi_{AE}$. To guess the user's outcome, she could then perform some measurement $\mathbbm{Z}=\{Z_e\}$, where $e$ denotes the outcome of Eve, on her part of the joint quantum state $\psi_{AE}$. The probability of Eve obtaining an outcome $e=a$ given the user's outcome $a$ is denoted as $p(e=a|a,\mathbbm{Z})$. Since Eve does not have access to the outcome $a$, the guessing probability of Eve is averaged over the outcomes of the user giving us the following expression
\begin{eqnarray}\label{guess1}
     p_{guess}(E|S)=\max_\mathbbm{Z}\sum_{\mathbf{a}}p(a)p(e=a|a,\mathbbm{Z})
\end{eqnarray}
where $S$ denotes the system of the user and $\mathbf{a}=a_1,a_2,\ldots,a_N$. For a note, the above formula is inspired from randomness generation in the Bell scenario \cite{random0}. Now, expressing \eqref{guess1} in quantum theory, we obtain that
\begin{equation}\label{guess3}
    p_{guess}(E|S)=\max_\mathbbm{Z}\sum_{\mathbf{a}}\Tr\left(\Pi_{x_{1},a_{1}}\Pi_{x',a'}\Pi_{x_{1},a_{1}}\otimes Z_{\mathbf{a}}\psi_{AE}\right)
\end{equation}
where
\begin{equation}
   \Pi_{x',a'}= \Pi_{x_{2},a_{2}}\ldots\Pi_{x_{N-1},a_{N-1}}\Pi_{x_N,a_N}\Pi_{x_{N-1},a_{N-1}}\ldots\Pi_{x_{2},a_{2}}.
\end{equation}
The projectors $\Pi_{x_i,a_i}$ are certified from  Eq. \eqref{lem1.2}  as $\Pi_{x_i,a_i}=\mathcal{U}^{\dagger}\left(\proj{e_{x_i,a_i}}\otimes\I\right)\mathcal{U}$, where $\ket{e_{x_i,a_i}}$ are the eigenstates of $\tilde{\mathcal{A}}_{x_i}$ [see Eq. \eqref{obs1}]. Thus, the guessing probability from Eq. \eqref{guess3} can be simplified to
\begin{eqnarray}\label{guess3}
     p_{guess}(E|S)=\qquad\qquad\qquad\qquad\qquad\qquad\qquad\nonumber\\ \max_\mathbbm{Z}\sum_{\mathbf{a}}\mathcal{N}_{\mathbf{a}}\ \Tr\left(\proj{e_{x_1,a_i}}\otimes\I_{A''}\otimes Z_{\mathbf{a}}\psi'_{AE}\right)
\end{eqnarray}
where $\psi'_{AE}=\mathcal{U}\psi'_{AE}\mathcal{U}^{\dagger}$ with
\begin{eqnarray}
\mathcal{N}_{\mathbf{a}}=\prod_{l=1}^{N-1}|\langle{e_{x_l,a_i}}\ket{e_{x_{l+1},a_i}}|^2.
\end{eqnarray}
Now, choosing $x_1=2,x_2=2+n/2,x_3=2,x_4=2+n/2\ldots$ we obtain that $\mathcal{N}_{a}=\frac{1}{2^{N-1}}$ for any $a$. Thus, the expression \eqref{guess3} is further simplified to
\begin{eqnarray}\label{guess4}
     p_{guess}(E|S)=\qquad\qquad\qquad\qquad\qquad\qquad\qquad\nonumber\\\frac{1}{2^{N-1}} \max_\mathbbm{Z}\sum_{\mathbf{a}}\ \Tr\left(\proj{e_{x_1,a_i}}\otimes\I_{A''}\otimes Z_{\mathbf{a}}\psi'_{AE}\right)
\end{eqnarray}
As the observable $\mathcal{A}_{x_1}$ acts on $\mathbbm{C}^2\otimes\mathcal{H}_{A''}$, we express the state $\ket{\psi_{AE}}$ as
\begin{eqnarray}\label{genstateeve}
\ket{\psi_{AE}'}=\sum_{i=0,1}\lambda_{a_i}\ket{e_{x_1,a_i}}_{A'}\ket{f_{i}}_{A''E}.
\end{eqnarray}
such that $\sum_{i=0,1}\lambda_{a_i}^2=1$ and the states $\ket{f_{i}}_{A''E}$ are in general not othogonal.
Plugging the above state Eq. \eqref{genstateeve} into Eq. \eqref{guess4} gives us
\begin{eqnarray}
    p_{guess}(E|S)=\frac{1}{2^{N-1}} \max_\mathbbm{Z}\sum_{\mathbf{a}}\lambda_{a_i}^2\ \bra{f_i}\I_{A''}\otimes Z_{\mathbf{a}}\ket{f_i}.
\end{eqnarray}
Using the fact that $\sum_{\mathbf{a}}Z_{\mathbf{a}}=\I$, we finally obtain  that
\begin{eqnarray}
     p_{guess}(E|S)=\frac{1}{2^{N-1}}.
\end{eqnarray}
The amount of randomness that can be extracted is quantified by the min-entropy of Eve's guessing probability \cite{di4}. Consequently, we obtain $N-1$ bits of randomness from $N-$sequential measurements. 
In principle, $N$ can be arbitrarily large and thus we can obtain an unbounded amount of randomness. Let us stress here that one can also obtain unbounded randomness when $n$ is odd. However, the amount of randomness obtained with $N-$sequential measurements is lower when $n$ is odd than even. It is also important to note here that one needs to input $2\log_2n$ bits of randomness in the scheme for the LG test. So in the proposed scheme, the first two measurements of the $N-$sequence need to be freely chosen. After this, it is not required as even if Eve knows the inputs she can not guess the outcomes. 

Let us notice that in the above protocol of randomness certification, we only considered the LG scenario with an even number of measurements. However, it can also be straightaway extended to the scenario with an odd number of measurements. However, in that case, one would obtain less than $N-1$ bits from $N$ sequential measurements. The reason is that the post-measured states corresponding to any measurements in the odd LG scenario would not give completely random outputs for any of the certified measurements. Consequently, for each of Alice's inputs, Eve can guess the outcomes with more than $1/2$ probability but strictly less than $1$.

Analysing from a phenomenological perspective, even if Eve has maliciously prepared an entangled source such that it sends a part of the state to her, the first projective measurement will break the entanglement and then Eve would have no connection with the state of Alice. Consequently, even if Eve knows the inputs or the measurements of Alice she can not guess the outcomes as there are no shared resources between her and Alice. This is why Eve can perfectly guess the first measurement outcome in the sequence but cannot guess any more of the outcomes in sequence with more than $1/2$ probability. Consequently, we obtain $N-1$ bits of secure randomness from $N$ sequential measurements.

{\it{Discussions---}}
In the scenario considered in this work, all the operations of the device occur locally where the device might have access to the previous inputs and outputs. 
For instance, the device might already have a list of instructions conditioned on the previous input and output in a stochastic way and build up the observed statistics. This possibility can never be excluded unless one finds some physical constraint such that the device does not store the information of the previous input and output. In the device-independent scenario, this possibility is excluded due to the space-like separation that does not allow one side to gain information about the other side. 
Consequently, as discussed above, we consider the assumption of "input-consistent measurements" \ref{ass1} which allows us to exclude the possibility of a classically pre-programmed device. Let us stress that such an assumption is natural in space-like separated scenarios but is an enforced assumption for the time-like separated scenario considered in this work. However, apart from device-independent ones, in every other quantum experiment, one naturally assumes that the correlations are generated by some measurement acting on some state and these measurements remain the same throughout the experiment. As pointed out by the referee, a few semi-device-independent schemes are also able to close this loophole \cite{Zhang2021,Nie_2024}.

Compared to semi-device independent randomness generation, our protocol is more secure as the assumption of "input-consistent measurements" is more natural than considering trusted measurements (source-independent scenario) \cite{SI1,SI2,SI3,SI4} or the dimension (prepare and measure scenario) \cite{PMrand1, PMrand2,PMrand3,PMrand4}. It is clear that trusting measurements is much stronger than assuming that the measurements remain consistent throughout the experiment. Trusting dimension, although weaker than trusting measurements, might allow an adversary to generate fake randomness by coupling an additional system with the input states that remain hidden from the user. Most importantly, our scheme can be implemented by using just some noise in the system, unlike any other known randomness generation scheme, where one needs to prepare specific states. In Appendix E of \cite{SupMat}, we also provide a possible protocol that can be easily implemented. As the source can in principle be any noise, one can even utilise microwave background radiation to generate this randomness.

Several follow-up problems arise from our work. 
An interesting problem would be to find the robustness of our protocol towards experimental imperfections. Further on, it would be highly desirable to generalise the above scheme to arbitrary number of outcomes to generate an arbitrary amount of randomness from a single measurement in a state-independent way. It would also be highly desirable if one can self-test any qubit measurement in a single experiment using the above scheme. 

\begin{acknowledgements}
We would like to thank Stefano Pironio for useful insights. This project was funded within the QuantERA II Programme (VERIqTAS project) that has received funding from the European Union’s Horizon 2020 research and innovation programme under Grant Agreement No 101017733.
    
\end{acknowledgements}

\providecommand{\noopsort}[1]{}\providecommand{\singleletter}[1]{#1}%

\onecolumngrid
\appendix
\section{Projectivity of quantum measurements}


\begin{fakt}\label{fact1}
    Assume that in the sequential scenario depicted in Fig. 1 of the manuscript, the correlations $\vec{p}_2$ are generated via input-consistent measurements $A_i$ acting on some quantum state $\rho_A$ [see assumption 1 of the manuscript]. Then the Zeno conditions \eqref{zenocond} implies that the measurements $A_i$ are projective.
\end{fakt}
\begin{proof}
    To begin with, let us expand the condition \eqref{zenocond} for $i=1$ using the L\"{u}der's rule to obtain the following expression
    \begin{equation}
         \Tr\left(\sqrt{\mathbbm{M}_{a}}\ U_{a}^{\dagger}\mathbbm{M}_{b}U_{a}\sqrt{\mathbbm{M}_{a}}\ \rho_A\right)= \delta_{a,b}\Tr\left(\mathbbm{M}_{a}\ \rho_A\right)
    \end{equation}
    where for simplicity we dropped the index $i=1$.
    Let us consider the case when $a\ne b$ in the above expression to obtain the following condition
    \begin{eqnarray}
        \Tr\left(\sqrt{\mathbbm{M}_{a}}\ U_{a}^{\dagger}\mathbbm{M}_{b}U_{a}\sqrt{\mathbbm{M}_{a}}\ \rho_A\right)=\qquad\nonumber\\||\sqrt{\mathbbm{M}_{b}}U_{a}\sqrt{\mathbbm{M}_{a}}\ \sqrt{\rho_A}||=0.
         \end{eqnarray}
    It is straightforward to conclude from the above expression that
    \begin{eqnarray}
 \sqrt{\mathbbm{M}_{b}}U_{a}\sqrt{\mathbbm{M}_{a}}\ \sqrt{\rho_A}=0
    \end{eqnarray}
    which on utilising the fact that $\rho_A$ is full-rank and thus invertible, we obtain
    \begin{eqnarray}\label{A3}
       \sqrt{\mathbbm{M}_{b}}U_{a}\sqrt{\mathbbm{M}_{a}}=0. 
    \end{eqnarray}
    Now multiplying $\sqrt{\mathbbm{M}_{b}}$ from left-hand side and $\sqrt{\mathbbm{M}_{a}}$ from right-hand side and using the fact that $\sum_a\mathbbm{M}_{a}=\I$, we obtain that
    \begin{eqnarray}\label{A4}
        U_a\mathbbm{M}_{a}=\mathbbm{M}_{a}U_a\mathbbm{M}_{a},\qquad a=0,1.
    \end{eqnarray}
    Let us now expand $\mathbbm{M}_{a}$ using its eigendecomposition as
    \begin{eqnarray}\label{ma}
        \mathbbm{M}_{a}=\sum_{k}\lambda_{k,a}\proj{e_{k,a}}
    \end{eqnarray}
    where $0\leq\lambda_{k,j}\leq 1$ and $\{\ket{e_{k,a}}\}_k$ are orthonormal set of vectors for any $a$. Let us also observe that $U_a\mathbbm{M}_{a}=\sum_{k}\lambda_{k,a}\ket{f_{k,a}}\!\bra{e_{k,a}}$ where $\ket{f_{k,a}}=U_a\ket{e_{k,a}}$. Consequently, we obtain from Eq. \eqref{A4} that 
    \begin{eqnarray}\label{A5}
       \sum_{k}\lambda_{k,a}\ \ket{f_{k,a}}\!\bra{e_{k,a}}=\sum_{l,k}\lambda_{l,a}\lambda_{k,a}\ \ket{e_{l,a}}\!\langle e_{l,a}\ket{f_{k,a}}\!\bra{e_{k,a}}.
    \end{eqnarray}
    Sandwiching the above expression with $\bra{e_{l,a}}..\ket{e_{k,a}}$ gives us
    \begin{eqnarray}\label{A6}
        \lambda_{k,a}\ \langle e_{l,a}\ket{f_{k,a}}=\lambda_{l,a}\lambda_{k,a}\ \langle e_{l,a}\ket{f_{k,a}}\qquad \forall l,k.
    \end{eqnarray}
    There exist atleast one $k$ for each $l$ such that $\langle e_{l,a}\ket{f_{k,a}}\ne0$ or else the condition Eq. \eqref{A5} can not be satisfied. Thus, we obtain from Eq. \eqref{A6} that $\lambda_{l,a}=1$ for all $l,a$. Thus, the measurement $\mathbbm{M}_a$ from Eq. \eqref{ma} is projective.
\end{proof}
\section{Some mathematical fact}
\begin{fakt}\label{fact1}
If $\alpha_i=\frac{\sin\left(\frac{\pi i}{n}\right)}{\sin\left(\frac{\pi (i+1)}{n}\right)}$ and $\beta_i=\frac{\sin\left(\frac{\pi}{n}\right)}{\sin\left(\frac{\pi (i+1)}{n}\right)}$, then 
\begin{eqnarray}\label{eq1}
    \sum_{i=1}^{n-2} \left(\frac{1}{\alpha_i}+\alpha_i+\frac{\beta_i^2}{\alpha_i}\right)=2n\cos\left(\frac{\pi}{n}\right).
\end{eqnarray}
\begin{proof}
    Let us first expand the term $t_i=\frac{1}{\alpha_i}+\alpha_i+\frac{\beta_i^2}{\alpha_i}$ for any $i$, 
    \begin{equation}\label{eq2}
        t_i=\frac{\sin\left(\frac{\pi (i+1)}{n}\right)}{\sin\left(\frac{\pi i}{n}\right)}+\frac{\sin\left(\frac{\pi i}{n}\right)}{\sin\left(\frac{\pi (i+1)}{n}\right)}+\frac{\sin^2\left(\frac{\pi}{n}\right)}{\sin\left(\frac{\pi (i+1)}{n}\right)\sin\left(\frac{\pi i}{n}\right)}.
    \end{equation}
Using the identity $\sin(a+b)=\sin(a)\cos(b)+\sin(b)\cos(a)$, we obtain from Eq. \eqref{eq2} that
\begin{eqnarray}\label{eq4}
t_i=2\cos\left(\frac{\pi}{n}\right)+\sin\left(\frac{\pi}{n}\right)\left[\cot\left(\frac{\pi i}{n}\right)-\cot\left(\frac{\pi(i+1)}{n}\right)\right]\nonumber\\ +\frac{\sin^2\left(\frac{\pi}{n}\right)}{\sin\left(\frac{\pi (i+1)}{n}\right)\sin\left(\frac{\pi i}{n}\right)}.\qquad
\end{eqnarray}
Now, expressing 
\begin{eqnarray}
    \sin\left(\frac{\pi}{n}\right)=\sin\left(\frac{\pi (i+1)}{n}-\frac{\pi i}{n}\right)
\end{eqnarray}
and again using the identity $\sin(a+b)=\sin(a)\cos(b)+\sin(b)\cos(a)$, we obtain from Eq. \eqref{eq4} 
\begin{equation}
    t_i=2\cos\left(\frac{\pi}{n}\right)+2\sin\left(\frac{\pi}{n}\right)\left[\cot\left(\frac{\pi i}{n}\right)-\cot\left(\frac{\pi(i+1)}{n}\right)\right]
\end{equation}
Now, summing $t_i$ over $i$ gives us
\begin{eqnarray}
    \sum_{i=1}^{n-2}t_i&=&2(n-2)\cos\left(\frac{\pi}{n}\right)+4\sin\left(\frac{\pi}{n}\right)\cot\left(\frac{\pi}{n}\right)\nonumber\\&=&2n\cos\left(\frac{\pi}{n}\right).
\end{eqnarray}
This completes the proof.
\end{proof}

\end{fakt}

\section{Self-testing the measurements}
Let us begin by recalling the LG functional
\begin{eqnarray}\label{lg2}
     \mathcal{L}&=&\frac{1}{2}\sum_{x=1}^{n-1}\left\langle\{\mathcal{A}_x,\mathcal{A}_{x+1}\}\right\rangle-\frac{1}{2}\left\langle\{\mathcal{A}_n,\mathcal{A}_1\}\right\rangle.
\end{eqnarray}
Consider now the following observables
\begin{eqnarray}\label{obs1}
    \tilde{\mathcal{A}}_x=\cos{\frac{\pi (x-1)}{n}}\sigma_z+\sin{\frac{\pi (x-1)}{n}}\sigma_x
\end{eqnarray}
where $\sigma_z, \sigma_x$ are the Pauli $z,x$ matrices. Then, one obtains the maximal quantum value of \eqref{lg2} to be $\beta_Q(n)=n\cos{\frac{\pi}{n}}$.
Consider now the following operators $P_i$ for $i=1,\ldots,n-2$ given by
\begin{eqnarray}\label{P}
    P_i=\mathcal{A}_i-\alpha_ i \mathcal{A}_{i+1}+\beta_ i \mathcal{A}_{n}
\end{eqnarray}
where 
\begin{eqnarray}\label{alpha}
    \alpha_i=\frac{\sin\left(\frac{\pi i}{n}\right)}{\sin\left(\frac{\pi (i+1)}{n}\right)},\qquad \beta_i=\frac{\sin\left(\frac{\pi}{n}\right)}{\sin\left(\frac{\pi (i+1)}{n}\right)}.
\end{eqnarray}
We now observe that
\begin{equation}\label{sos1}
    \sum_{i=1}^{n-2}\frac{1}{2\alpha_i}P_i^{\dagger}P_i=\sum_{i=1}^{n-2}\frac{1}{2\alpha_i}\left[(1+\alpha_i^2+\beta_i^2)\I-\alpha_i\{\mathcal{A}_i,\mathcal{A}_{i+1}\}+\beta_i\{\mathcal{A}_i,\mathcal{A}_{n}\}-\alpha_i\beta_i\{\mathcal{A}_n,\mathcal{A}_{i+1}\}\right]=\beta_Q(n)\I-\hat{\mathcal{L}}
\end{equation}
where we used the fact that $\mathcal{A}_i^2=\I$ and $\alpha_{i+1}\beta_i=\beta_{i+1}$.

Now, let us assume that one observes the value $ \beta_Q(n)$ of the LG functional $\mathcal{L}$ \eqref{lg2}. Thus from the decomposition \eqref{sos1}, we have that 
\begin{eqnarray}\label{sosrel2}
    \Tr(P_i^{\dagger}P_i\rho_A)=0, \qquad i=1,\ldots,n-2.
\end{eqnarray}
The above relation \eqref{sosrel2} will be particularly useful for self-testing as stated below.

\setcounter{thm}{0}
\begin{thm}\label{Theo1M} 
Assume that the Zeno conditions \eqref{zenocond} are satisfied and the LG inequality \eqref{lg2} is maximally violated by some state $\rho_A$ and observables $\mathcal{A}_i\ (i=1,\ldots,n)$. Then, the following statements hold true:
\\
\\
1. The observables $\mathcal{A}_i$ act on the Hilbert space  $\mathcal{H}_{A}=(\mathbbm{C}^2)_{A'}\otimes \mathcal{H}_{A''}$ for some auxiliary Hilbert space $\mathcal{H}_{A''}$.\\
\\
2.  \ \  There exist unitary transformations, $\mathcal{U}:\mathcal{H}_A\rightarrow\mathcal{H}_A$,  such that
\begin{eqnarray}\label{lem1.2}
\mathcal{U}\mathcal{A}_i\mathcal{U}^{\dagger}=\tilde{\mathcal{A}}_i\otimes\mathbbm{1}_{A''}.
\end{eqnarray}
where the observables $\tilde{\mathcal{A}}_i$ are listed in Eq. \eqref{obs1}.
\end{thm}
\begin{proof}

Let us begin by considering the relation \eqref{sosrel2} which can be rewritten as $||P_i\sqrt{\rho_A}||=0$ for $i=1,\ldots,n-2$ and thus we obtain that $ P_i\sqrt{\rho_A}=0$. As $\rho_A$ is full-rank, we simply arrive at the condition $P_i=0$ which can be expanded using \eqref{P} to obtain
\begin{eqnarray}\label{sosrel5}
   \mathcal{A}_i=\alpha_ i \mathcal{A}_{i+1}-\beta_ i \mathcal{A}_{n}\qquad i=1,\ldots,n-2.
\end{eqnarray}
Let us now consider $i=1$ in the above formula \eqref{sosrel5} and substitute $\alpha_1,\beta_1$ from Eq. \eqref{alpha} to arrive at
\begin{eqnarray}\label{sosrel6}
     \mathcal{A}_1=\frac{1}{2\cos\left(\frac{\pi}{n}\right)}( \mathcal{A}_{2}-\mathcal{A}_{n}).
\end{eqnarray}
Again using the fact that $ \mathcal{A}_i^2=\I$, allows us to conclude from the above formula \eqref{sosrel6}
\begin{eqnarray}
    \frac{1}{4\cos^2\left(\frac{\pi}{n}\right)}( \mathcal{A}_{2}-\mathcal{A}_{n})^2=\I
\end{eqnarray}
which on further expansion gives us
\begin{eqnarray}\label{sosrel7}
    \{\mathcal{A}_2,\mathcal{A}_n\}=2\left[1-2\cos^2\left(\frac{\pi}{n}\right)\right]\I.
\end{eqnarray}
Let us now show that the observables $\mathcal{A}_i$ for any $i$ are traceless. For this purpose, we consider the above formula \eqref{sosrel7} and multiply it with $\mathcal{A}_2$ and then take the trace to obtain
\begin{eqnarray}\label{tr1}
    \Tr(\mathcal{A}_n)=\left[1-2\cos^2\left(\frac{\pi}{n}\right)\right]\Tr(\mathcal{A}_2).
\end{eqnarray}
Again, we consider Eq. \eqref{sosrel7} and multiply it with $\mathcal{A}_n$ and then take the trace to obtain
\begin{eqnarray}\label{tr2}
    \Tr(\mathcal{A}_2)=\left[1-2\cos^2\left(\frac{\pi}{n}\right)\right]\Tr(\mathcal{A}_n).
\end{eqnarray}
It is straightforward from Eqs. \eqref{tr1} and \eqref{tr2}, that $\Tr(\mathcal{A}_2)=\Tr(\mathcal{A}_n)=0$ for any $n\geq3$. Further on, taking trace on both sides of Eq. \eqref{sosrel5} for any $i$, allows us to conclude that $\Tr(\mathcal{A}_i)=0$. Thus, the number of eigenvalues $(1,-1)$ of the observables $\mathcal{A}_i $ are equal. Consequently, the  observables $\mathcal{A}_i $ act on $\mathbbm{C}^2\otimes\mathcal{H}_{A''}$.

Let us now characterize the observables $\mathcal{A}_i$. For this purpose, we observe from \eqref{sosrel7} that
\begin{eqnarray}\label{sosrel11}
     \frac{1}{4\sin^2\left(\frac{\pi}{n}\right)}( \mathcal{A}_{2}+\mathcal{A}_{n})^2=\I.
\end{eqnarray}
Let us further notice that $\{\mathcal{A}_{2}-\mathcal{A}_{n},\mathcal{A}_{2}+\mathcal{A}_{n}\}=0$ which can rewritten as
\begin{equation}\label{sosrel10}
    \left\{\frac{1}{2\cos\left(\frac{\pi}{n}\right)}( \mathcal{A}_{2}-\mathcal{A}_{n}),\frac{1}{2\sin\left(\frac{\pi}{n}\right)}( \mathcal{A}_{2}+\mathcal{A}_{n})\right\}=0.
\end{equation}
As proven in \cite{Jed1}, if two matrices $A,B$ anti-commute and $A^2=B^2=\I$, then there exist a unitary transformation $\mathcal{U}$ such that $\mathcal{U}A\mathcal{U}^{\dagger}=\sigma_z\otimes\I$ and $\mathcal{U}B\mathcal{U}^{\dagger}=\sigma_x\otimes\I$. Thus, from Eqs. \eqref{sosrel6}, \eqref{sosrel11} and \eqref{sosrel10} we obtain that
\begin{eqnarray}\label{sosrel12}
     \mathcal{A}'_{2}-\mathcal{A}'_{n}&=&2\cos\left(\frac{\pi}{n}\right)\sigma_z\otimes\I,\nonumber\\  \mathcal{A}'_{2}+\mathcal{A}'_{n}&=&2\sin\left(\frac{\pi}{n}\right)\sigma_x\otimes\I
\end{eqnarray}
where $\mathcal{A}'_{i}=\mathcal{U}\mathcal{A}_{i}\mathcal{U}^{\dagger}$. Thus, we obtain from Eqs. \eqref{sosrel6} and \eqref{sosrel12} that
\begin{eqnarray}\label{sosrel15}
    \mathcal{A}'_{1}&=&\sigma_z\otimes\I\nonumber\\
    \mathcal{A}'_{2}&=&\left(\cos{\frac{\pi}{n}}\sigma_z+\sin{\frac{\pi}{n}}\sigma_x\right)\otimes\I\nonumber\\
    \mathcal{A}'_{n}&=&\left(-\cos{\frac{\pi}{n}}\sigma_z+\sin{\frac{\pi}{n}}\sigma_x\right)\otimes\I.
\end{eqnarray}
Now, let us consider the condition \eqref{sosrel5} for $i=2$ and apply $\mathcal{U}$ on both the sides to obtain
\begin{eqnarray}
    \mathcal{A}'_2=\alpha_ 2 \mathcal{A}'_{3}-\beta_ 2 \mathcal{A}'_{n}.
\end{eqnarray}
Now, substituting $\alpha_2,\beta_2$ from Eq. \eqref{alpha} and $\mathcal{A}'_2,\mathcal{A}'_n$ from \eqref{sosrel15} and then after some trigonometric simplification, we obtain
\begin{eqnarray}
    \mathcal{A}'_{3}=\left(\cos{\frac{2\pi}{n}}\sigma_z+\sin{\frac{2\pi}{n}}\sigma_x\right)\otimes\I.
\end{eqnarray}
Continuing in a similar fashion for all $i's$ allows us to conclude that for $i=1,2,\ldots,n$
\begin{eqnarray}
    \mathcal{A}_i'=\left(\cos{\frac{\pi (i-1)}{n}}\sigma_z+\sin{\frac{\pi (i-1)}{n}}\sigma_x\right)\otimes\I.
\end{eqnarray}
This completes the proof.
\end{proof}

\section{Robustness to experimental errors}

\begin{thm}
    Suppose that the observables in the actual experiment are close to the ideal ones as
    \begin{eqnarray}\label{obserr}
    ||(\mathcal{A}_i-\mathcal{A}'_i)\sqrt{\rho}||\leq \varepsilon
    \end{eqnarray}
where $\mathcal{A}'_i=\mathcal{U}\left(\tilde{\mathcal{A}}_i\otimes\I\right) \mathcal{U}^{\dagger}$ and $\tilde{\mathcal{A}}_i$ are listed in Eq. \eqref{obs1} with $\mathcal{A}_i$ being projective. Here $\rho$ is the actual state during the experiment. Then, the LG inequality \eqref{lg2} is violated close to the quantum bound as
\begin{eqnarray}
    \mathcal{L}\geq \beta_Q(n)-\frac{n(1+2\cos(\pi/n))}{2}\varepsilon.
\end{eqnarray}
\end{thm}

\begin{proof}
 To begin with, let us consider the sum of squares decomposition of the LG inequality \eqref{sos1} and rewrite it as
 \begin{eqnarray}\label{robu2}
    \mathcal{L}=\Tr \left(\hat{\mathcal{L}}\rho\right)= -\sum_{i=1}^{N-2}\frac{1}{2\alpha_i}||P_i\sqrt{\rho}||+\frac{1}{2}\sum_{i=1}^{n-2}\left(\frac{1}{\alpha_i}||\mathcal{A}_i\sqrt{\rho}||\right.\nonumber\\ \left.+\alpha_i||\mathcal{A}_{i+1}\sqrt{\rho}||+\frac{\beta_i^2}{\alpha_i}||\mathcal{A}_n\sqrt{\rho}||\right).\quad
 \end{eqnarray}
 To find the lower bound to $ \mathcal{L}$, we find the lower bound to $||\mathcal{A}_i\sqrt{\rho}||$ for $i=1,\ldots,n$ and upper bound to $||P_i\sqrt{\rho}||$ for $i=1,\ldots,n-2$.

    Let us first find the lower bound of $||\mathcal{A}_i\sqrt{\rho}||$ for all $i$. For this purpose, we consider the expression \eqref{obserr} and expand it using the identity: $||a|-|b||\leq|a-b|$ to obtain
    \begin{eqnarray}\label{robu1}
         -\varepsilon\leq||\mathcal{A}_i\sqrt{\rho}||-||\mathcal{A}'_i\sqrt{\rho}||\leq \varepsilon.
    \end{eqnarray}
    As $\mathcal{A}'_i$ is unitary for any $i$ [see Eq. \eqref{obs1}] and consequently $||\mathcal{A}'_i\sqrt{\rho}||=1$, we obtain from \eqref{robu1} that
    \begin{eqnarray}\label{robu3}
        ||\mathcal{A}_i\sqrt{\rho}||\geq1-\varepsilon.
    \end{eqnarray}
    
   Let us now find the upper bound to $||P_i\sqrt{\rho}||$ for all $i$. For this purpose, let us first observe that 
    \begin{eqnarray}
\tilde{\mathcal{A}}_i=\alpha_ i \tilde{\mathcal{A}}_{i+1}-\beta_ i \tilde{\mathcal{A}}_{n}
    \end{eqnarray}
    where $\tilde{\mathcal{A}}_i$ are the ideal observables listed in Eq. \eqref{obs1} and $\alpha_ i,\beta_i$ are given in Eq. \eqref{alpha}. Now, it is simple to observe from Eq. \eqref{P} that
    \begin{eqnarray}
        ||P_i\sqrt{\rho}||=||(\mathcal{A}_i-\mathcal{A}_i')\sqrt{\rho}-\alpha_ i (\mathcal{A}_{i+1}-\mathcal{A}_{i+1}')\sqrt{\rho}\nonumber\\+\beta_ i (\mathcal{A}_{n}-\mathcal{A}_{n}')\sqrt{\rho}||.
    \end{eqnarray}
    Now using triangle inequality, we obtain that
    \begin{eqnarray}
         ||P_i\sqrt{\rho}||\leq||(\mathcal{A}_i-\mathcal{A}_i')\sqrt{\rho}||+\alpha_ i ||(\mathcal{A}_{i+1}-\mathcal{A}_{i+1}')\sqrt{\rho}||\nonumber\\+\beta_ i ||(\mathcal{A}_{n}-\mathcal{A}_{n}')\sqrt{\rho}||.\quad
    \end{eqnarray}
    which utilising \eqref{obserr} gives us
    \begin{eqnarray}\label{robu4}
        ||P_i\sqrt{\rho}||\leq(1+\alpha_i+\beta_i)\varepsilon\qquad\forall i.
    \end{eqnarray}
Thus, from Eqs. \eqref{robu2}, \eqref{robu3} and $\eqref{robu4}$ we obtain that
\begin{eqnarray}
    \mathcal{L}\geq\beta_Q(n)-\sum_{i=1}^{N-2}\frac{(1+\alpha_i+\beta_i)}{2\alpha_i}\varepsilon.
\end{eqnarray}

\end{proof}

\section{A possible protocol for implementation}

Here we present a possible protocol for implementing the randomness generation scheme in an optical setup. For simplicity, we consider the sequential scenario [see Fig. 1 of the manuscript] when the number of measurements $n=4$. Let us stress that we do not consider all the practical constraints that might affect the experiment but present it from a more theoretical standpoint.

\begin{itemize}
    \item {\bf{Source.}} The source is prepared by the user. As there does not need to be any control on the source even sending some thermal light into the device is sufficient.
    \item {\bf{Measurements.}} The measurement could be the simple optical implementation of the measurements $\{Z,X,(X-Z)/\sqrt{2},(X+Z)/\sqrt{2}\}$. For instance, one can follow the approach of \cite{LGexp3}.
    \item {\bf{Parameter estimation.}} In some rounds of the experiment, the user has to estimate the value of the Leggett-Garg functional $\mathcal{L}$ \eqref{lg2}. For this purpose, the user needs to input $4$ bits of randomness for each round of the estimation. This comes from the fact that in each round of parameter estimation, one has to freely choose two inputs for evaluating $\mathcal{L}$. 
     \item {\bf{Randomness extraction.}} In all the other rounds, (or even in the rounds of the parameter estimation), the incoming signal should be measured sequentially as long as the signal can be detected by the measurement devices. If the signal can be measured sequentially for $N$ times, then one can obtain $N-1$ bits of certified genuine randomness from each round of the experiment. 
\end{itemize}

\end{document}